\documentclass{article}
\usepackage{graphicx} 

\usepackage{amsmath,amssymb,amsthm,yhmath,amsfonts,mathrsfs,enumitem,mathtools,dsfont}
\usepackage[utf8]{inputenc}
\usepackage{tikz}
\usetikzlibrary{shapes,arrows,positioning}

\newcommand{\W}{\mathbb{W}}
\renewcommand{\d}{\mathrm{d}}

\newcommand{\1}[1]{{\mathds{1}}_{\left\{#1\right\}}}

\usepackage{graphicx}
\usepackage{tikz}
\usetikzlibrary{automata,positioning}
\usepackage{url}

\newtheorem{theorem}{Theorem}[section]
\newtheorem{definition}[theorem]{Definition}
\newtheorem{lemma}[theorem]{Lemma}

\newtheorem{remark}[theorem]{Remark}

\newtheorem{conjecture}[theorem]{Conjecture}

\usepackage{authblk}

\title{Why is it easier to predict the epidemic curve than to reconstruct the underlying contact network?}
\author[1,2]{Dániel Keliger\footnote{Dániel Keliger is partially supported by the ERC Synergy under Grant No. 810115 - DYNASNET,  EK\"OP-24-3-BME-316  of the Ministry
for Culture and Innovation from the source of the National Research Development and Innovation
Fund and NKFI-FK-142124.  \includegraphics[width=3cm]{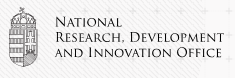}}}
\author[3,4]{Illés Horváth\footnote{Illés Horváth is partially supported by the OTKA K-138208 project.}}
\affil[1]{Department of Stochastics, Budapest University of Technology and Economics, Hungary}
\affil[2]{Alfréd Rényi Institute of Mathematics, Budapest, Hungary}
\affil[3]{Department of Networked Systems and Services,
Budapest University of Technology and Economics, Hungary}
\affil[4]{HUN-REN-BME Information Systems Research Group}

\begin{document}

\maketitle

\begin{abstract}
    We study the deterministic Susceptible-Infected-Susceptible (SIS) epidemic model on weighted graphs. In their numerical study \cite{reconstruction_mieghem2022} van Mieghem et al. have shown that it is possible to learn an estimated network from a finite time sample of the trajectories of the dynamics that in turn can give an accurate prediction beyond the sample time range, even though the estimated network might be qualitatively far from the ground truth. We give a mathematically rigorous derivation for this phenomenon, notably that for large networks, prediction of the epidemic curves is robust, while reconstructing the underlying network is ill-conditioned. Furthermore, we also provide an explicit formula for the underlying network when reconstruction is possible. At the heart of the explanation, we rely on Szemer\'edi's weak regularity lemma.
\end{abstract}

\section{Introduction}

One of the main approach of epidemic modeling is compartment-based system of differential equations where the compartments correspond to geographic locations or age groups in different stages of the infection  \cite{santiago}. Besides the initial conditions, infection and cure rates, arguably, the most challenging parameter to estimate is the underlying contact network on which the spreading phenomenon occurs. Direct methods usually rely on making modeling assumptions or using proxy networks based on social media networks, GPS and mobile data, commuting patterns or questioners \cite{mobile_epidemic, switchover2021}, methods which can be potentially intrusive in terms of privacy or give an unfaithful representation.

An alternative paradigm treats network estimation as a learning problem where we first use observations up to some time $T_0>0$ to find the parameters that are most consistent with the epidemic curves thus far, then using said parameters to extrapolate the curve into a future time $T>T_0$ \cite{reconstruction_mieghem2022}. Surprisingly, while this method yields an estimated network capable of making accurate predictions for a reasonably long time scale, the realized estimate might radically differ from the ground truth alluding to the fact that only partial information of the underlying network is essential in the evolution of the pandemic. In summary, they show that while the \emph{reconstruction} of the network is hard, \emph{predicting} the epidemic curves is robust. 

In \cite{reconstruction_mieghem2022} the phenomena was demonstrated numerically along with some heuristic arguments. In \cite{Beaufort2022} the authors study a slightly different setting where instead of commuting (individuals returning to their home vertex at a given rate \cite{santiago}), individuals perform random walks according to a diffusion matrix. They rigorously prove that almost every diffusion matrix from almost every initial condition can be reconstructed from the trajectories and numerically show that said reconstruction can be challenging. Reconstruction of the background network graphon based on the distribution of branching processes was examined in \cite{hladky2024}. \cite{reconstruction_entropy} provides a different setup with a stochastic process on a random graph, with an entropy inequality implying a tradeoff between prediction and reconstruction. Intuitively, if the evolution of the pandemic is sensitive to the details of the network, then one can extract a rich picture from the trajectories, while prediction becomes hard as even slight perturbations of the network results in a drastic change of dynamic. On the other hand, if many realizations of the random graph give similar trajectories, extrapolation becomes more feasible.

In the present paper, we follow the setting of \cite{reconstruction_mieghem2022} for the Susceptible-Infected-Susceptible (SIS) model and give a proper mathematical setup and rigorous proof that shows that the prediction problem is robust, while the reconstruction problem is ill-conditioned for large networks. To the best of our knowledge, these results are the first in this endeavor. Furthermore, we also give an explicit formula for the underlying network using only observable quantities for cases when the network can be uniquely reconstructed.

We believe the concepts and methods are robust for generalization to other kind of models; an important notion is the observable quantities, and irregularities of the evolution also potentially play a role. These could be examined on a case-by-case basis; in the present paper, we do not pursue a general framework.

The rest of the paper is structured as follows: In Section \ref{s:setup} we give an introduction of the SIS model and rigorously define the reconstruction and prediction problems. In Section \ref{s:reuslts} we state the main results. Section \ref{s:graphon} introduces concepts and notations regarding graphons \cite{lovaszbook} which are crucial for the proofs but non-essential in terms of stating the problems and the corresponding theorems. Section \ref{s:proofs} contains the proofs for the main results.

\section{Setup}
\label{s:setup}
In this section we give a precise formulation of the model and the problems of interest.

\subsection{The SIS model}

We study the deterministic SIS metapopulation model. For simplicity of presentation, we choose to focus on SIS model, but most of the results can be extended to a wider class of epidemic models including the SIR or SEIR processes with appropriate modifications (e.g. to the set of observable quantities).

Assume there are $n$ communities represented by the vertex set $V= \{1, \dots , n \}$. The graph is equipped with both symmetric edge weights $w_{ij}=w_{ji}\geq 0$ forming the matrix $W:= \left( w_{ij}\right)_{i,j \in V}$. Symmetry is usually assumed for similar epidemic models, but is actually not necessary for most of the calculations.

Positive vertex weights $\pi_{i} \geq 0$ are also given with normalization
\begin{align*}
\sum_{i=1}^{n} \pi_i =1.
\end{align*}
The vector $\pi:= \left( \pi_{i} \right)_{i \in V}$ is referred to as the distribution of the population between communities. We also use the matrix notation $\Pi:= \operatorname{diag} \left( \pi \right)$. 

$w_{ij}$ represents the strength of interaction between communities $i,j \in V$. It can also be proportional to the edge density if the weighted graph corresponds to a coarse grained version of a simple graph. Any global constants (e.g. infection rate usually denoted by $\beta$) can also be absorbed in the notation.

$\gamma_i \geq 0$ denotes the curing rate of community $i \in V$; the corresponding vector notation is $\gamma:= \left(\gamma_i \right)_{i \in V}.$ Inhomogeneity of the curing rate might be an important modeling aspect when communities correspond to age groups or other biological characteristics of the population that can affect the curing rate.

The special case $\gamma=0$ is referred to as the SI model. 

In the SIS model, individuals can occupy one of two states: susceptible and infected. Let $z_i(t)$ stand for the ratio of infected individuals at community $i \in V$ at time $t \geq 0.$ The dynamics is given by the following system of ODEs:
\begin{align}
\label{eq:NIMFA}
\frac{\d}{\d t} z_{i}(t)=(1-z_{i}(t)) \sum_{j \in V} w_{ij} \pi_{j} z_{j}(t)-\gamma_{i} z_{i}(t).
\end{align}

Further explanation of \eqref{eq:NIMFA}. The term $-\gamma_{i} z_{i}(t)$ corresponds to the curing of infected individuals in community $i \in V$. $z_i(t)$ is increased when a susceptible individual from community $i$ makes an interaction with an infected individual from some community $j$. Assuming that the communities are well-mixed, the probability that the vertex from community $i$ is susceptible and the vertex from community $j$ is infected is $(1-z_i(t))z_{j}(t)$. Finally, $w_{ij} \pi_{j}$ is proportional to the rate at which such interaction accrues regardless of the state of individuals. 

Using the notations $\mathcal{A}:=W \Pi= \left(a_{ij} \right)_{i,j \in V}$ and $z(t):= \left(z_{i}(t) \right)_{i \in V}$, we may rewrite \eqref{eq:NIMFA} as
\begin{align}
\label{eq:NIMFA_A}
\frac{\d}{\d t} z_{i}(t)=(1-z_i(t)) \left(\mathcal{A} z(t) \right)_{i}-\gamma_i z_{i}(t)=:F_i^{\mathcal{A}}(z(t)) \qquad (i\in V).
\end{align}

The \emph{degree} of community $i \in V$ is defined as
\begin{align}
\label{eq:delta}
d_i:= \sum_{j \in V} a_{ij}
\end{align}
representing the maximal rate at which an individual at community $i \in V$ can get infected.

The system \eqref{eq:NIMFA} is intimately related to the stochastic version of the SIS process on sufficiently dense simple graphs. Fix the parameters $n, W, \pi, \gamma,$ and generate a random graph on $N$ vertices where two vertices from communities $i,j \in V$ respectively are connected with probability (proportional to) $\kappa^N w_{ij}$, where $\kappa^N$ is a density parameter such that $\kappa^N \gg \frac{ \log N}{ N}.$ Then, according to \cite{Keliger2022}, the ratio of infected in community $i \in V$ will stochastically converge to $z_{i}(t)$ uniformly on any compact time interval $[0,T].$ (Note that although the setting in \cite{Keliger2022} does not allow inhomogeneous $\gamma_i$ parameters directly, one can augment the state space to consist of pairs $(i,S), (i,I)$ for different $i \in V$.  Since $n$ is fixed, the results in \cite{Keliger2022} remain valid for this state space as well.) 

Since the right hand side of \eqref{eq:NIMFA} is locally Lipschitz-continuous, \eqref{eq:NIMFA} has a unique solution. Furthermore, standard arguments show that the region $[0,1]^n$ is positive invariant \cite{SISdyn}, that is,
$$0 \leq z_i(0) \leq 1 \quad \forall i \in V \quad \Longrightarrow \quad 0 \leq z_i(t) \leq 1 \quad \forall t \geq 0 \,\forall i \in V,$$
so $z_i(t)$ can indeed be interpreted as a ratio.

We point out that based on the Cauchy–Kovalevskaya theorem $z_i(t)$ is also an analytic function of $t$ as the right hand side of \eqref{eq:NIMFA} is itself analytic. This property will be frequently exploited during the proofs, starting with the following lemma. 
\begin{lemma}
\label{l:analytic}
Let $h : [0,\infty [ \to \mathbb{R}$ be an analytic function. Assume for some $T_0>0$ we have $\forall t \in [0,T_0] \ h(t)=0$.  Then, $\forall t \geq 0 \ h(t)=0.$  
\end{lemma}

We will make use of the Gram matrix
\begin{align}
\label{eq:Gdef}
G_{ij}(T):= \int_{0}^{T}z_{i}(t) z_{j}(t) \d t.
\end{align}
It is symmetric, positive semi-definite, hence it induces a semi-scalar product
\begin{align*}
\|v \|_{G(T)}^2:= \langle v, G(T) v \rangle=\int_{0}^{T} \langle v, z(t) \rangle^2 \d t,
\end{align*}
where $\langle \cdot,\cdot\rangle$ is the standard scalar product on $\mathbb{R}^n$.

$\|v \|_{G(T)}$ is zero on the nullspace of $G(T)$, denoted by $\mathcal{N}.$ Lemma \ref{l:analytic} applied for the analytic function $h(t):=\langle v, z(t) \rangle $ implies that for any $0<T_0 \leq T$ the nullspaces are the same, hence, $\mathcal{N}$ does not depend on $T$ as long as $T>0$.

We will also use the $L^r(V,\pi)$ norm:
\begin{align*}
 \|v \|_{r,\pi}:= \left (\sum_{i \in V} \left| v_i \right|^r \pi_i\right)^{1/r}. 
\end{align*}

Besides the induced matrix norm $\| \cdot \|_{p, \pi}$, we will also use the so called cut norm defined as
\begin{align}
\label{eq:cut_norm_finite}
\|W \|_{\square}:= \max_{S,S' \subseteq V} \left | \sum_{ \substack{i \in S \\ j \in S'}}w_{ij }\pi_i \pi_j \right|.
\end{align}

\subsection{The reconstruction problem}

The setup for the reconstruction problem is the following. We have perfect knowledge of $z(t)$ over a time interval $[0,T_0]$ ($T_0>0$), along with $\pi$ and $\gamma.$ The task is to find $W$. 

In practice, $\pi$ being available is a reasonable assumption since it only concerns the size of communities. Knowledge about $\gamma$ can be interpreted in several ways: it could be experimentally measured during clinical trials, or a different approach would be to assume that we do not only observe the total change of $z(t)$ over time, but we also observe the gains $\left( \frac{\d}{\d t} z_{i}(t) \right)_{+}:= (1-z_{i}(t)) \left(\mathcal{A}z(t) \right)_{i}$ and losses $\left(\frac{\d}{\d t} z_{i}(t) \right)_{-}:=-\gamma_{i}z_{i}(t)$ separately, corresponding to new infections and recoveries respectively. From these, $\gamma_i$ can be expressed as
\begin{align*}
\gamma_{i}=- \frac{1}{z_{i}(t)} \left( \frac{\d }{\d t}z_{i}(t) \right)_{-}.    
\end{align*}

We note that $\mathcal{A}z(t)$ is also an observable quantity since it has the form of
\begin{align}
\label{eq:A_observe}
\left(\mathcal{A}z(t) \right)_{i}= \frac{\frac{\d}{\d t}z(t)+\gamma_iz_{i}(t)}{1-z_{i}(t)}.
\end{align}

We set some technical restrictions on the parameters. In general, we want the weights $w_{ij}$ and the curing rates $\gamma_i$ to be bounded from above, and $z_i(t)$ to be bounded away from $1$ (as $1-z_i(t)$ appears in the denominator of the right hand side of \eqref{eq:A_observe}). We also want $z_i(0), \pi_i$ and $w_{ij}$ to be bounded away from 0 to ensure a lower bound on the speed of the disease spreading (especially for the SI case).

Accordingly, for bounds $0 <m \leq 1/2$ and $0<M$, we define the parameter sets
\begin{align}
\label{eq:Q}
\begin{split}
\mathcal{Q}_{m,M}^{(n)}:=&  \Bigg\{(z(0),\gamma,\pi) \in \mathbb{R}^n \times \mathbb{R}^n \times \mathbb{R}^n  | \forall i \in V \ m \leq z_{i}(0) \leq 1-m, \\
& 0 \leq \gamma_i \leq M, 0\leq \pi_i, \sum_{j\in V} \pi_j =1 \Bigg\},
\end{split}\\
\label{eq:P}
\mathcal{P}_{m,M}^{(n)}:=& \left \{ (W,q) \in \mathbb{R}^{n \times n} \times \mathcal{Q}_{m,M}^{(n)}  \big| W^{\intercal}=W,\, \forall i,j \in V \, m \leq w_{ij} \leq M \right \}.
\end{align}
\begin{remark}
For a fixed $n$, the assumption $m \leq w_{ij} \leq M$ is not too restrictive, but later on it will be used as $n \to \infty$ with $0<m \leq M$ fixed, essentially making the graph sequence dense, enabling us to apply the theory of graphons in section \ref{s:graphon}. See also the proofs of Theorems \ref{t:reconstruction_bad} and \ref{t:uniform_prediction}.
\end{remark}

For the reconstruction problem, we will always assume $\pi_i>0\, \forall i \in V$ as the dynamics is insensitive to vertices with $\pi_i=0$ (so reconstruction for those parts of the graph would be impossible anyway). $\pi_i=0$ is not an issue for the prediction problem though, where it is convenient to allow $\pi_i=0$ as $n$ is allowed to increase.

Since \eqref{eq:NIMFA} is well-posed, the parameters $ p \in \mathcal{P}_{m,M}^{(n)} $ uniquely determine the trajectories $ \left(z(t) \right)_{0 \leq t \leq T_0}$, hence we may define the mapping
\begin{align*}
\phi_{T_0}(p): \ p \mapsto \left(z(t) \right)_{0 \leq t \leq T_0}.
\end{align*}

\begin{definition}
 For $p=(W, q) \in  \mathcal{P}_{m,M}^{(n)}$,  we say $W$ is \emph{$q$-reconstructible} if for all $T_0>0$, the mapping $\phi_{T_0} \left( \cdot, q \right)$ is injective at $W$.
\end{definition}

Reconstruction is impossible when $T_0=0$ as $W$ can be chosen independently of $q$. Hence, we assume the process is observable on a possibly small, but positive length interval. 

Not all $W$ are reconstructible. Assume that the parameter $p$ is regular in the sense that there exist  $\bar{\gamma}, \bar{d}, \bar{z}(0)$ such that for all $i \in V$ we have $d_i=\bar{d}, \ \gamma_i=\bar{\gamma}$ and $z_{i}(0)=\bar{z}(0).$ It is easy to check that for all $i \in V$ the solution has the form of $z_i(t)=\bar{z}(t)$ where $\bar{z}(t)$ solves
\begin{align*}
\frac{\d}{\d t}\bar{z}(t)=\bar{d} (1-\bar{z}(t))  \bar{z}(t)-\bar{\gamma}\bar{z}(t).
\end{align*}
However, many other set of parameters $p \in \mathcal{P}_{m,M}^{(n)}$ satisfy the regularity assumption with the same triple $(\bar{d}, \bar{\gamma}, \bar{z}(0))$.

\subsection{The prediction problem}

The prediction problem is the following: given $q$ and the curve $\left(z(t) \right)_{0 \leq t \leq T_0}$, we aim to predict $z(t)$ for $t>T_0$.

One approach to the prediction problem is to find an approximate weighted graph $\hat{W}$ such that the corresponding solution $\hat{z}(t)$ is close to the original solution $z(t)$ on $[0,T_0]$ according to some error measure, and use $\hat{z}(t)$ as a prediction for $z(t)$ for $t>T_0$. In \cite{reconstruction_mieghem2022} the authors constructed a $\hat{W}$ that is consistent with $z(t)$ on $[0,T_0]$ by minimizing the error between $F^{\mathcal{A}}(z(t))$ and $F^{\hat{\mathcal{A}}}(z(t)).$ More precisely, they solve
\begin{align*}
& \operatorname{argmin}_{\hat{a}_{i,1}, \dots \hat{a}_{i,n}} \sum_{j=1}^{J} \left(  \frac{\Delta z_i(t_j)}{\Delta t_j}-F_i^{\hat{\mathcal{A}}}(z(t_j)) \right)^2+\rho_i \sum_{j=1}^{J} \hat{a}_{ij} \\
& s.t. \ \hat{a}_{ij} \geq 0
\end{align*}
for all $i \in V.$ $\rho_i \geq 0$ serve as parameters to tune the relaxation terms.

We take a similar approach, but use a different error measure; instead of using the derivatives $F^{\mathcal{A}}(z(t))$ we make use of the neighborhood statistics $\mathcal{A}z(t)$ and $\hat{\mathcal{A}}z(t)$ ($\hat{\mathcal{A}}:=\hat{W} \Pi$) as these are also observable and have nicer analytical properties. 

We introduce the $L^r\left(V, \pi \right)$ type error terms:
\begin{align}
\label{eq:H}
H_{r,T}(W,W'):= \left(\int_{0}^{T} \|\mathcal{A}z(t)-\hat{\mathcal{A}} z(t) \|_{r,\pi}^r \d t \right)^{1/r}.
\end{align}
$H_{r,T_0}$ is an observable quantity but $H_{r,T}$ is not as it depends on the future values of $z(t).$

The following lemma states that if the error $H_{2,T}$ is small, then the trajectories generated by $W$ and $\hat{W}$ stay close together on $[0,T].$
\begin{lemma}
\label{l:H1}
Let $p, \hat{p} \in \mathcal{P}_{0,M}^{(n)}$ with the corresponding solutions of \eqref{eq:NIMFA} being $z(t)$ and $\hat{z}(t).$

\begin{align*}
\sup_{0 \leq t \leq T} \|z(t)-\hat{z}(t) \|_{1,\pi} \leq& \left(\|z(0)-\hat{z}(0) \|_{1,\pi}+T \|\gamma-\hat{\gamma} \|+H_{1,T} \right)e^{3MT}.
\end{align*}
Furthermore, assuming $\hat{z}(0)=z(0)$ and $\hat{\gamma}=\gamma$, we also have
\begin{align*}
   \sup_{0 \leq t \leq T} \|z(t)-\hat{z}(t) \|_{1,\pi} \leq e^{3MT}H_{1,T} \leq \sqrt{T}e^{3MT} H_{2,T}. 
\end{align*}
\end{lemma}

Let $\eta^{(i)}$ denote the $i$th row vector of $\mathcal{A}-\hat{\mathcal{A}}$ (that is, $\eta^{(i)}_j:= a_{ij}-\hat{a}_{ij}).$ Then
\begin{align*}
  \| \left(\mathcal{A}-\hat{\mathcal{A}} \right) z(t) \|_{2,\pi}^2=\sum_{i \in V} \langle \eta^{(i)}, z(t) \rangle^2 \pi_i,  
\end{align*}
implying
\begin{align}
\label{eq:H2_noice}
H_{2,T}^2=\sum_{i \in V} \left \| \eta^{(i)} \right \|_{G(T)}^2 \pi_i.
\end{align}

The minimum of $H_{2,T_0}=H_{2,T_0}(\hat{W})$ as a function of $\hat{W}$ is 0, obtained at $\hat{W}=W.$ Using Lemma \ref{l:analytic} for the function $h(t)=H_{2,t},$ we obtain $H_{2,T_0}=0 \ \Rightarrow H_{2,T}=0.$ This means that if we manage to find the minimum of $H_{2,T_0}$ without error, then perfect prediction is possible.

However, this is rarely possible; typically, we target $H_{2,T_0}$ to be smaller than some given threshold, and the task is then to show continuous dependence of the future error $H_{2,T}$ on the past error $H_{2,T_0}$, that is, we want to show that for all $\varepsilon >0$ there is a $\delta>0$ such that $H_{2,T_0}<\delta  \Rightarrow H_{2,T}<\varepsilon.$

\section{Results}
\label{s:reuslts}

\subsection{Results for the reconstruction problem}
\label{s:results_reconstruction}

We start by giving a characterization of $q$-reconstructability.
\begin{theorem}
\label{t:reconstruction_characterisation}
Let $p=(W,q) \in \mathcal{P}_{m,M}^{(n)}$ and assume $\pi_i>0\, \forall i \in V.$ 
 
The following properties are equivalent:

 \begin{enumerate}
     \item\label{condrecona} $W$ is $q$-reconstructible.
     \item\label{condreconb} For any $T_0>0$ and $v \in \mathbb{R}^n$ $\forall t \in [0,T_0] \ \langle v, z(t) \rangle \Rightarrow v=0. $
     \item\label{condreconc} For any $T_0>0$, $G(T_0)$ is invertible.
 \end{enumerate}

\end{theorem}

$W$ not being $q$-reconstructible does not necessarily mean that the counterexample $\hat{W}$ which generates the same trajectory also satisfies $(\hat{W},q)\in \mathcal{P}_{m,M}^{(n)}.$ That said, in the proof of Theorem \ref{t:reconstruction_characterisation}, the $\hat{W}$ constructed satisfies $\hat{W} \sim \mathcal{P}_{m',M'}$ for arbitrary $0<m'<m \leq M <M'$. 


Reconstructability can also be understood in the literal sense, that is, there is an explicit formula for $W$. First define
\begin{align*}
    R(T_0):= \int_{0}^{T_0} \mathcal{A} z(t) z^{\intercal}(t) \d t.
\end{align*}
Both $G(T_0)$ (defined in \eqref{eq:Gdef}) and $R(T_0)$ are observable matrices which satisfy the identity
\begin{align}
 \label{eq:A_explicit}
 \mathcal{A}G(T_0)=R(T_0).
\end{align}
According to Theorem \ref{t:reconstruction_characterisation}, if $W$ is $q$-reconstructible, then $G(T_0)$ is invertible, and
\begin{align}
\label{eq:W_explicit}
    W=R(T_0)G^{-1}(T_0) \Pi^{-1}.
\end{align}

For any $\hat{\mathcal{A}}=\hat{W} \Pi$ satisfying $\hat{\mathcal{A}}G(T_0)=R(T_0)$, we also have
\begin{align*}
H_{2,T_0}^2(\hat{W})=&\int_{0}^{T_0} \sum_{i \in V} \left( \left(\mathcal{A}-\hat{\mathcal{A}} \right)z(t) \right)_i^2 \pi_i \d t\\
=&\sum_{i  \in V} \pi_i  \Bigg(\underbrace{\int_{0}^{T_0}  \left(\mathcal{A}-\hat{\mathcal{A}} \right)z(t) z^{\intercal}(t) \d t }_{=R(T_0)-\hat{\mathcal{A}}G(T_0)=0}\left(\mathcal{A}-\hat{\mathcal{A}} \right)^{\intercal} \Bigg)_{ii}=0, 
\end{align*}
therefore, even if  the solution to \eqref{eq:A_explicit} is not unique, any other solution gives the same trajectory as $W$. 

None of the conditions of Theorem \ref{t:reconstruction_characterisation} are easy to check directly for the parameter $p$. We conjecture that there exists an equivalent condition that can be checked directly. In order to present it, we define block regularity.

\begin{definition}
  We call $p \in \mathcal{P}_{m,M}^{(n)}$ \emph{block regular} if there is a partition $V_1,\dots, V_K$ of $V$ for some $K<n$ such that
  \begin{align*}
    &\forall 
    1 \leq k \leq K \  i_1,i_2 \in V_k \ z_{i_1}(0)=z_{i_2}(0), \ \gamma_{i_1}=\gamma_{i_2 }, \\
    & \forall 1 \leq k,l \leq K \ i_1, i_2 \in V_k,  \ \sum_{j \in V_l}a_{i_1j}=\sum_{j \in V_l} a_{i_2j}.
  \end{align*}
\end{definition}

\begin{conjecture}
\label{c:regular}
Assuming $\forall i \in V \ \pi_i>0$ and  $p=(W, q) \in \mathcal{P}_{m,M}^{(n)}$ $W$ is $q$-reconstructible if and only if $p$ is not block regular.     
\end{conjecture}

The only if part of Conjecture \ref{c:regular} is relatively straightforward. Block regular graphs can be reduced to a graph on $K$ vertices. Fix $i \in V_k$. Define $\bar{z}_{k}(0):=z_{i}(0), \ \bar{\gamma}_k:=\gamma_i, \ \bar{\pi}_{k}:=\sum_{j \in V_k} \pi_j$ and 
\begin{align*}
\bar{a}_{kl:}=\sum_{i \in V_k} \frac{\pi_i}{\bar{\pi}_k} \sum_{i \in V_K} \sum_{j \in V_l} a_{ij}.
\end{align*}
Then
\begin{align*}
  \bar{w}_{kl}:= \frac{1}{\bar{\pi}_k \bar{\pi}_l}\sum_{i \in V_k}\sum_{j \in V_l} w_{ij}\pi_i \pi_j  
\end{align*}
is symmetric and satisfies $\bar{w}_{kl}\bar{\pi}_l=\bar{a}_{kl}.$ 

Clearly $\bar{p}:=\left(\bar{W}, \bar{z}(0),  \bar{\gamma}, \bar{\pi} \right) \in \mathcal{P}_{m,M}^{(K)}.$ Furthermore, the solution $z(t)$ has the following form:
$$\forall 1 \leq k \leq K, \ i \in V_K, t \geq 0 \ z_{i}(t)=\bar{z}_k(t),$$
where $\bar{z}(t)$ solves \eqref{eq:NIMFA} with parameter $\bar{p}.$


Choose $i_1,i_2\in V_k$, $i_1\neq i_2$, and let
\begin{align*}
    v_j= 
    \left\{\begin{array}{cl}
        1  & \text{if } j=i_1, \\
        -1 & \text{if } j=i_2, \\
        0, & \text{otherwise.}
    \end{array}\right.
\end{align*}
Then clearly $\langle v, z(t)\rangle=z_{i_1}(t)-z_{i_2}(t)=0,$ so according to Condition \ref{condreconb} of Theorem \ref{t:reconstruction_characterisation}, $W$ is non-$q$-reconstructible.




For the special case of the SI processes, there is a neat sufficient condition for reconstructability.

\begin{theorem}
\label{t:different_d}
For $q \in \mathcal{Q}_{m,M}^{(n)}$, assume the following properties:
\begin{itemize}
\item $\gamma=0$ (SI model), and
\item $\forall i \in V \ \pi_i>0$, and
\item the degrees $d_1, \dots, d_n$ are all different.
\end{itemize}
Then $W$ is $q$-reconstructible.
\end{theorem}

\begin{remark}
If $\pi$ is uniform and $W$ represents a simple unweighted graph, then there exist two equal degrees as the possible degrees are $0,1,\dots,n-1$, but $0$ and $n-1$ cannot be present simultaneously.
\end{remark}

\begin{remark}
For any fixed $q \in \mathcal{Q}_{m,M}^{(n)}$ with $\gamma=0$ and $ \forall i \in V, \pi_i>0$, the degrees are all different for almost every $W$, because $d_i=d_j$ ($i \neq j$) gives a lower dimensional subspace for $W$.
\end{remark}

Lemma 6 in \cite{Beaufort2022} presents a result similar to Theorem \ref{t:different_d}, with some differences. \cite{Beaufort2022} studies the SIR process with diffusion and takes finitely many observations as opposed to full knowledge of the process over $[0,T_0].$ Most importantly though, the proof in \cite{Beaufort2022} relies on looking only at the total number of individuals at a given location, which only depends on the diffusion process but not on the epidemic dynamics, so the result is more about how to reconstruct the diffusion coefficients from the trajectories of the diffusion process rather than from the epidemic. The proof is also not applicable to the case when the diffusion is in equilibrium (even though that would be a natural assumption), while the result in the present paper is applicable for arbitrary initial conditions within $Q^{(n)}_{m,M}$.

\begin{remark}
    In the SI model, the degrees are observable in the $t\to\infty$ limit due to
    \begin{align*}
    \lim_{t \to \infty} -\frac{1}{t} \log \left(1- z_i(t) \right)=\lim_{t \to \infty} \frac{1}{t}\int_{0}^{t} \left( \mathcal{A}z(\tau)\right)_i \d \tau-\frac{1}{t} \log \left(1- z_{i}(0) \right)= \left( \mathcal{A} \mathds{1} \right)_i=d_i.
    \end{align*}

    The relevance of this identity is that in \cite{reconstruction_mieghem2022}, one way to show that $\hat{W}$ significantly differs from $W$ was to compare their degree distribution. However -- at least for the SI case -- if the error between $z(t)$ and $\hat{z}(t)$ converges to $0$, then the degree distribution of $\hat W$ also converges to the degree distribution of $W$.
\end{remark}

While Theorem \ref{t:different_d} states that for any given $q$, almost every $W$ is $q$-reconstructible (at least in the SI case), it does not imply that the reconstruction problem is robust. Indeed, if we take a sequence $p_N \to p$ where $p=(W,q)$ is such that $W$ is non-$q$-reconstructible, then \eqref{eq:A_explicit} becomes more and more ill-conditioned as $N \to \infty$.

The reconstruction problem becomes harder as we let $n \to \infty$. We introduce a further definition: a special case of block regular parameters are block homogeneous ones.
\begin{definition}
$p \in \mathcal{P}_{m,M}^{(n)}$ is \emph{block homogeneous} if it is block regular and it also satisfies
\begin{align*}
  \forall 1 \leq k, l \leq K, i \in V_k, j \in V_l \ w_{ij}=\bar{w}_{kl},
\end{align*}
and $m \leq \bar{w}_{kl} \leq M \,\forall 1 \leq k, l \leq K.$
\end{definition}

Szemer\'edi's regularity lemma essentially states that any $p \in \mathcal{P}^{(n)}_{m,M}$ can be approximated by some block homogeneous $\hat{p} \in \mathcal{P}_{m,M}^{(n)}$ corresponding to some lower dimensional parameter $\bar{p} \in \mathcal{P}_{m,M}^{(K)}$ where $K \ll n.$

\begin{theorem}
\label{t:reconstruction_bad}
 For every $\varepsilon>0$ there is an $n_0$ such that for all $n \geq n_0$ and $p=(W,q) \in \mathcal{P}_{m,M}^{(n)}$ with $\pi_i=\frac{1}{n}$ there is a $\hat{W}$ such that
 \begin{align*}
   &\|W-\hat{W} \|_{\square} \geq  \frac{1}{16}m>0  \\
   & \sup_{0 \leq t \leq T} \|z(t)-\hat{z}(t)\|_{1, \pi}<\varepsilon.
 \end{align*}
\end{theorem}

\begin{remark}
The condition $\pi_i= \frac{1}{n}$ is mainly for convenience. However, some regularity about the population size must be assumed. As a counterexample, consider $\pi_{1}=1-\frac{1}{n-1}$ and $\pi_{i}=\frac{1}{(n-1)^2}$ for $i>1$. In this setting, the problem is basically reduced to a weighted graph with one vertex where the only relevant parameter from $W$ is $w_{11}$ corresponding to the infection rate of the main block. 
\end{remark}

\subsection{Results for the prediction problem}
\label{s:results_prediction}

 Let $v \in \mathbb{R}^n$ be arbitrary and $u \in \mathcal{N}^{\intercal}$ be its orthogonal projection. Clearly, $\| v\|_{G(T)}=\|u \|_{G(T)}$ for all $T>0.$ Notice that $\mathcal{N}^{\intercal}$ is a finite dimensional subspace on which both $\| \cdot \|_{G(T)}$ and  $\| \cdot \|_{G(T_0)}$ are norms. Since on finite dimensional spaces all norms are equivalent, there is a constant $\Lambda^{(n)}(p)$ such that
 \begin{align*}
\|u \|_{G(T)} \leq \Lambda^{(n)}(p) \| u\|_{G(T_0)}\quad \forall u \in \mathcal{N}^{\intercal},
 \end{align*}
implying
\begin{align}
\label{eq:Lambda}
 H_{2,T} \leq \Lambda^{(n)}(p)  H_{2,T_0}  
\end{align}
based on \eqref{eq:H2_noice}. (Note that $\Lambda^{(n)}(p)$ also depends on $T_0,T,m,M$  and $\pi$.)

This means that for any \emph{fixed} $p \in \mathcal{P}_{m,M}^{(n)}$ prediction is possible, and we also have $H_{2,T}=O \left( H_{2,T_0} \right).$

However, \eqref{eq:Lambda} is weak in the sense that the dependence on $p$ is not specified. For example, the previous argument does not exclude the possibility that for some sequence $p_N \to p^*$ ($N \to \infty$) we have $\Lambda^{(n)}\left(p_N\right) \to \infty $ even though $\Lambda^{(n)}(p^*)<\infty$. ($\Lambda^{(n)}(p)$ is not continuous in $p$.)

We conjecture this not to be the case for the prediction problem.

\begin{conjecture}
\begin{align}
\label{eq:Lambda_bound1}
\Lambda^{(n)}:=&\sup_{p \in \mathcal{P}^{(n)}_{m,M}} \Lambda^{(n)}(p)<\infty \\
\label{eq:Lambda_bound2}
\Lambda:=& \sup_{n} \Lambda^{(n)}<\infty
\end{align}    
\end{conjecture}

A slightly weaker yet still surprisingly strong statement show, however, that for $H_{2,T_0}<\delta \Rightarrow H_{2,T}<\varepsilon$ we can choose $\delta$ independently of both $\mathcal{P}_{0,M}^{(n)}$ and $n$. This means the prediction problem can be solved very robustly even for large dimension.

\begin{theorem}
\label{t:uniform_prediction}

Assuming $ 0 \leq \hat{w}_{ij} \leq M$ ($i,j \in V$) we have that
\begin{align}
 \forall \varepsilon>0 \ \exists \delta>0 \ \forall  p \in  \bigcup_{n \in \mathbb{N^{+}}}\mathcal{P}_{0,M}^{(n)} \ H_{2,T_0}<\delta \Rightarrow H_{2,T}<\varepsilon .
\end{align}
\end{theorem}

\subsection{Discrete time and measurement errors}
\label{s:discretemeasure}

So far, we assumed that $z(t)$ is available for all $0\leq t\leq T_0$ with arbitrary precision. In real-life scenarios, this is typically not the case; the values of $z(t)$ are only available at certain measurement points, and the $z(t_j)$ values may also have measurement errors. This section addresses the prediction error resulting from only knowing $z(t)$ at the discrete measurement points, and also the error resulting from measurement error.

The measurement points are denoted by
$0=t_0 < t_1 < \dots < t_J = T_0$, and the measurement values are $z^*(t_i), i=0,\dots, J.$ Each $z^*(t_i)$ is a vector of length $n$. Technically, $z^*(\cdot)$ is only defined at the $t_0,t_1,\dots,t_J$ points, but we will consider it extended to $[0,T_0]$ as a piecewise constant function between the measurement points. $\gamma$ is also measured with error; the corresponding measurement is denoted by $\gamma^*\in \mathbb{R}^n$. Since it is convenient to use the $L^2(V,\pi)$ norm, for the sake of simplicity, we do not account for the the measurement error of $\pi$.

We introduce the error terms
\begin{align*}
\Delta_1&:= \|z(0)-z^*(0)\|_{\pi,2},\\
\Delta_2&:= \|\gamma-\gamma^*\|_{\pi,2},\\
\Delta_3&:=  \max_{1\leq i\leq J}\frac{1}{t_i-t_{i-1}}
\|(z(t_i)-z(t_{i-1}))-(z^*(t_i)-z^*(t_{i-1}))\|_{\pi,2},\\
\Delta_4&:= \max_{1\leq i\leq J} (t_i-t_{i-1}).
\end{align*} 

$\Delta_3$ essentially corresponds to the measurement error of $\frac{\d z(t)}{\d t}$; it does not necessarily decrease if we increase the number of measurements, 
creating an interesting tension with $\Delta_4$.

To keep the derivations shorter, discrete measurement points and measurement error will be handled simultaneously:
$$\Delta:=\Delta_1+\Delta_2+\Delta_3+\Delta_4,$$
and we assume that we can make measurements frequently and accurately enough to make $\Delta$ small.

We assume $p \in \mathcal{P}_{m,M}^{(n)}$. In subsequent notation, constants only depending on $m,M$ and $T_0$ will be included in $O(\cdot)$.

\begin{lemma}
\label{l:ztpi2s}
\begin{align}
\|z(t)-z^*(t)\|_{\pi,2}=O(\Delta).
\end{align}
\end{lemma}

For $p \in \mathcal{P}_{m,M}^{(n)}$ we have
$$1-z_j(t)\leq me^{-MT_0} \quad \forall t\in [0,T_0];$$
along with Lemma \ref{l:ztpi2s}, this implies that there exists a $c_0(m,M,T_0)$ such that if $\Delta\leq c_0(m,M,T_0)$,
$$1-z_j^*(t)\geq m_0 \quad\forall t\in [0,T_0].$$ 
In the rest of this section we assume $\Delta\leq c_0(m,M,T_0)$, so $\frac{1}{1-z_j^*(t)}=O(1).$

Motivated by \eqref{eq:A_observe}, we define 
\begin{align}
\alpha_j(t_i):=\frac{\frac{z_j^*(t_i)-z_j^*(t_{i-1})}{t_i-t_{i-1}}+\gamma_j^*z_j^*(t_i)}{1-z_j^*(t_i)}
\end{align}
and
\begin{align}
\label{eq:H2sdef}
(H_{2,T_0}^*(\hat W))^2:=
\sum_{i=1}^J\sum_{j=1}^n |\alpha_j(t_i)-(\hat{\mathcal{A}}z^*(t_i))_j|^2\pi_j (t_i-t_{i-1})
\end{align}

The main result of this section is the following:
\begin{theorem}
Assume $0 \leq w_{ij}, \hat{w}_{ij} \leq M$. Then we have that

\label{t:numerror}
\begin{align}
\label{eq:numerror}
(H_{2,T_0}^*(\hat W))^2=H_{2,T_0}^2(\hat W)+O(\Delta).
\end{align}
\end{theorem}
(The $O(\cdot)$ depends on $m,M,T_0$, but not on $n$.)

\section{Setup for graphons}
\label{s:graphon}
In scenarios where $n$ is allowed to vary, it is convenient to work in a state space that can incorporate weighted graphs of different sizes.

Let each vertex $i$ be represented by a subset $I_i \subseteq [0,1]$ with $\left| I_i \right|=\pi_i.$ such that $I_1, \dots I_n$ forms a partition of $[0,1]$. This way, we may lift the weighted graph $W= \left( w_{ij} \right)_{ij \in V}$ up to the space of symmetric (measurable) kernel functions $ m \leq W(x,y)=W(y,x) \leq M$ via
\begin{align*}
 W(x,y):=\sum_{i, j \in V}w_{ij} \1{x \in I_i} \1{y \in I_j },  
\end{align*}
or in other words, we set $W(x,y)=w_{ij}$ when $x \in I_i$ and $y\in I_j$.

Measurable kernels satisfying $0 \leq W(x,y) =W(y,x) \leq1$ (a.e. $x,y \in [0,1]$) are called graphons in the literature and they arise as the natural limit objects of dense graph sequences \cite{lovaszbook}. With a slight abuse of terminology, we will call any bounded, non-negative symmetric kernel graphons as well as they can always be rescaled.

The corresponding integral operator is defined as
\begin{align*}
  \W f(x):=\int_{0}^{1} W(x,y) f(y) \d y.  
\end{align*}

Similarly, the solution $(z_i(t))_{i \in V}$ is also lifted up to the space of (measurable) $\mathbb{R}_0^+ \times [0,1] \to [0,1] $ functions via
\begin{align*}
    u(t,x):=\sum_{i \in V}z_{i}(t) \1{x \in I_i},
\end{align*}
while the recovery rates $(\gamma_i)_{i \in V}$ become
\begin{align*}
\gamma(x):=\sum_{i \in V} \gamma_i \1{x \in I_i}.
\end{align*}

If $\left(z_i(t) \right)_{i \in V}$ is a solution to \eqref{eq:NIMFA}, then $u(t,x)$ is a solution to
\begin{align}
\label{eq:u}
\partial_tu(t,x)=(1-u(t,x)) \W u(t,x)-\gamma(x) u(t,x).
\end{align}

Define the extended parameter set to a set of measurable functions satisfying
\begin{align}
\begin{split}
\mathcal{P}_{m,M}:=& \big\{(W,u(0),\gamma) | m  \leq W(x,y) =W(y,x) \leq M,\\ & m \leq u(0,x) \leq 1-m,
0 \leq \gamma(x) \leq M, \textit{a.e } x,y \in [0,1]  \big\}.
\end{split}
\end{align}
For any parameter $p \in \mathcal{P}_{m,M}$, there is a solution $u(t)$ to \eqref{eq:u} defined for all $t \geq 0$ such that $0 \leq u(t,x) \leq 1$ (see \cite{SISdyn}).

We will call a parameter $p \in \mathcal{P}_{m,M}$ \emph{discrete} if it corresponds to some finite interval partition $I_1, \dots, I_n$ in which case $p$ can be identified as an element of $\mathcal{P}_{m,M}^{(n)}.$

The extension of \eqref{eq:H} for graphons becomes
\begin{align}
\label{eq:H_graphon}
H_{r,T}:= \left( \int_{0}^{T} \left \|Wu(t)-\hat{W}u(t) \right \|_r^r \d t \right)^{1/r},
\end{align}
which coincides with \eqref{eq:H} for the discrete case.

The cut norm of a graphon $W$ is defined as
\begin{align}
\label{eq:cut_norm_graphon}
\|W \|_{\square}:= \sup_{S,S' \subseteq [0,1]} \left|\int_{S} \int_{S'} W(x,y) \d x \d y \right|
\end{align}
where the supremum is taken over measurable sets.

\begin{remark}
   Assume $W$ is discrete and possibly negative. Let $S_i:=S \cap I_i, \ S_i':=S' \cap I_i$ and $x_i:=|S_i|, y_i:=|S_i'|$. Clearly, $0 \leq x_i,y_i \leq \pi_i.$
   
   Define
   \begin{align*}
      f(x,y):=x^{\intercal} Wy=\sum_{i,j} w_{ij}|S_i| |S'_j|=\int_{S} \int_{S'} W(x,y) \d x \d y. 
   \end{align*}
    Assume that $(x,y)$ either maximizes (or minimizes) $f$. Then $\frac{\partial f}{\partial x_i}= (Wy)_i=0$ for any $0<x_i<\pi_i$, implying
    \begin{align*}
        f(x,y)=x^{\intercal}Wy=\sum_{0<x_i<\pi_i} x_i (Wy)_i+\sum_{x_i=\pi_i}x_i (Wy)_i=\sum_{x_i=\pi_i}x_i (Wy)_i.
    \end{align*}
    This means that for any $i$ for which $0<x_i<\pi_i$, the same maximum (or minimum) value of $f$ can be obtained by setting $x_i=0$. Thus, without loss of generality, we may assume that $\forall i$ either $S_i=I_i$ or $S_i= \emptyset.$

    The same argument applies for $S'$ as well, hence, in the discrete case, \eqref{eq:cut_norm_graphon} and \eqref{eq:cut_norm_finite} are consistent.
\end{remark}

We extend the notion of the cut norm to the parameter set $\mathcal{P}_{m,M}$ as
\begin{align}
\label{eq:cut_norm_p}
\| p \|_{\square}:= \|W \|_{\square}+ \|u(0) \|_1+ \|\gamma \|_1.
\end{align}
Note that for $u(0)$ and $\gamma$, the $L^1([0,1])$ norm is used; in 1-dimension, it's equivalent to the 
cut norm.

A slight generalization of Lemma \ref{l:H1} yields
\begin{lemma}
\label{l:H1_graphon}
\begin{align*}
\sup_{0 \leq t \leq T} \|u(t)-\hat{u}(t) \|_1 \leq \left(\|u(0)-\hat{u}(0) \|_1+T\|\gamma-\hat{\gamma} \|_1+H_{1,T} \right)e^{3MT}.
\end{align*}
  
\end{lemma}

\begin{remark}
\label{eq:cut_norm_conv}
Take a sequence $(p_N) \subset \mathcal{P}_{m,M}$ such that $p_N \overset{\| \cdot \|_{\square}}{\to}p.$ This clearly implies $u_N(0) \overset{L^1([0,1])}{\to} u(0)$ and $ \gamma_N \overset{L^1([0,1])}{\to} \gamma.$

Define $V_N:=\frac{1}{M}(W_N-W)$ and the corresponding integral operator $\mathbb{V}_N:= \frac{1}{M}(\W_N-\W).$ Since $|V_n| \leq 1$, we may use the inequality
\begin{align*}
  \| \mathbb{V}_N \|_{2} \leq \sqrt{8 \|V_N \|_{\square} } \to 0
\end{align*}
from \cite{cutnorm_inequality}, implying
\begin{align*}
   H_{1,T}^2(W_N) \leq T H_{2,T}^2(W_N)=T M^2 \int_{0}^{T} \| \mathbb{V}_N u(t) \|_2^2  \d t \leq T^2 M^2 \| \mathbb{V}_N \|_2^2 \to 0,
\end{align*}
resulting in $u_N(t) \overset{L^1([0,1)]}{\to}u(t) $ uniformly in $t \in [0,T].$
\end{remark}

The graphon extension of Szemerédi's regularity lemma \cite{Szemeredi_analysis} shows that any graphon can be approximated up to some error $\varepsilon>0$ with a discrete graphon with at most $K_{\max}(\varepsilon)$ blocks. The proof can be extended to the parameter set $\mathcal{P}_{m,M}$ by minor changes.

\begin{lemma} (Szemerédi's lemma for $\mathcal{P}_{m,M}$)
\label{l:Szemeredi}

$ \forall \varepsilon>0 \ \exists K_{\max}(\varepsilon)\ \forall p \in \mathcal{P}_{m,M} \ \exists p' \in \mathcal{P}_{m,M}$ with the following properties:
\begin{itemize}
    \item[1)] $\|p-p' \|_{\square}<\varepsilon$,
    \item[2)] $p'$ is discrete with partition $I_1', I_2', \dots, I_K'$ such that $1 \leq K \leq K_{\max}(\varepsilon)$ and
    \begin{align*}
        w_{kl}':=&\frac{1}{\pi_k' \pi_l'} \int_{I_k} \int_{I_l}W(x,y) \d x \d y \\
        z'_k(0):=& \frac{1}{\pi_k} \int_{I_k} u(0,x) \d x \\
        \gamma_k':=& \frac{1}{\pi_k} \int_{I_k} \gamma(x) \d x,
    \end{align*}
    \item[3)] when $p$ is discrete with partition $I_1, \dots, I_n$ such that $n \geq K_{\max}(\varepsilon)$ then the partition $(I_k)_{k=1}^{K}$ can be chosen in a way  that  $(I_i)_{i=1}^n$ is a refinement of $(I_k)_{k=1}^{K}.$
\end{itemize}
\end{lemma}

\begin{remark}
    From 3) it is easy to see that $p'$ is block homogeneous when $p$ is discrete.
\end{remark}

One issue with the cut norm is that it depends on the labeling of the vertices and it might distinguish between two isomorphic graphs. Instead, we might want to consider unlabeled graphs where the permutation of vertices are treated as equivalent. The graphon analogy of permutations are bijective measure preserving mappings $\phi: [0,1] \to [0,1]$ forming the set $\Phi.$ 

The corresponding change of coordinates ("relabeling") of one and two variable functions $f(x), g(x,y)$ are denoted by
$$f^{\phi}(x):=f(\phi(x)),\qquad  g^{\phi}(x,y):=g(\phi(x),\phi(y)).$$ Similarly, for $p \in \mathcal{P}_{m,M}$, we use the convention $p^{\phi}:=(W^{\phi},(u(0))^{\phi},\gamma^{\phi})$.

It is easy to check that $\|p^{\phi} \|_{\square}=\| p\|_{\square}.$

\begin{remark}
    \label{r:u_relabel}
    Relabeling only effects the solution of \eqref{eq:u} up to a change of coordinates:
    \begin{align*}
    \partial_t u^{\phi}(t,x)=& \partial_t u(t,\phi(x))=(1-u(t,\phi(x))\int_{0}^{1}W(\phi(x),y) u(t,y) \d y \\
    &-\gamma(\phi(x))u(t,\phi(x)) \\
    =&(1-u^{\phi}(t,x))\int_{0}^{1}W^{\phi}(x,y)u^{\phi}(t,y) \d x-\gamma^{\phi}(x)u^{\phi}(t,x)
    \end{align*}
    corresponding to the solution of the parameter $p^{\phi}.$
\end{remark}

We define an equivalence relation between $p,p' \in \mathcal{P}_{m,M}$ where $p \sim p'$ if and only if there is a $\phi \in \Phi $ such that $p^{\phi}=p'$ a.e.  The factorized space is denoted by $\tilde{\mathcal{P}}_{m,M}:= \mathcal{P}_{m,M} / \sim$ on which we define the (analogue of the) cut distance:
\begin{align*}
d_{\square}(p,p'):=\inf_{ \phi \in \Phi} \left\| p^{\phi}-p' \right\|_{\square}.
\end{align*}

\begin{remark}
Clearly $d_{\square}(p,p') \geq 0$ with equality if and only if $p \sim p'.$ Symmetry also holds due to
$$\left\| p^{\phi}-p' \right\|_{\square}=\left\| p-\left(p'\right)^{\phi^{-1}} \right\|_{\square}.$$

For the triangle inequality, we note that for any $\varepsilon>0$ we may choose $\phi_1,\phi_2 \in \Phi$ such that
\begin{align*}
 \left\|p^{\phi_1}-p' \right\|_{\square} \leq d_{\square}(p,p')+\varepsilon, \qquad 
 \|p'-(p'')^{\phi_2} \|_{\square} \leq d_{\square}(p',p'')+\varepsilon.
\end{align*}
Then
\begin{align*}
d_{\square}(p,p'') & \leq \left \|p^{\phi_1 \circ \phi_2^{-1}} -p''\right \|_{\square}=\left \|p^{\phi_1 } -(p'')^{\phi_2}\right \|_{\square} \\
 &\leq \left\|p^{\phi_1}-p' \right\|_{\square}+\|p'-(p'')^{\phi_2} \|_{\square}  \leq d_{\square}(p,p')+d_{\square}(p',p'')+2 \varepsilon.
\end{align*}
Letting $\varepsilon\to 0$ gives
\begin{align*}
 d_{\square}(p,p'') & \leq d_{\square}(p,p')+d_{\square}(p',p'').
\end{align*}
\end{remark}

\begin{lemma}
\label{l:compact}
The metric space  $\left( \tilde{\mathcal{P}}_{m,M}, d_{\square} \right)$ is compact.
\end{lemma}

The martingale argument in the proof of Theorem 5.1 in \cite{Szemeredi_analysis} applies for the proof of Lemma \ref{l:compact} with straightforward modifications. 

\section{Proofs}
\label{s:proofs}
\subsection{Proofs of general statements}

\begin{proof} (Lemma \ref{l:analytic})
Define $T:= \sup \{t \geq 0 | h(t)=0 \}.$ Clearly, $0<T_0 \leq T$. Indirectly assume $T<\infty.$ 
For all $0 \leq t<T$ we have $h(t)=0$, which means all the left derivatives of $h$ are $0$ at $T$, hence, the Taylor series expansion is $0$ around $T$. Then there exists an $\varepsilon>0$ such that $h(t)=0 \, \forall t \in [T,T+\varepsilon].$ This would imply $T +\varepsilon \leq T $, resulting in a contradiction.
\end{proof}

Next, we prove Lemma \ref{l:H1_graphon} which directly implies Lemma \ref{l:H1}.

\begin{proof} (Lemma \ref{l:H1_graphon})
\begin{align*}
u(t,x)=& \, u(0,x)+\int_{0}^{t} \left(1-u(\tau,x) \right)\W u(t,x)-\gamma(x) u(\tau,x) \d \tau \\
\left |u(t,x)-\hat{u}(t,x) \right| \leq & \left |u(0,x)-\hat{u}(0,x) \right| +\int_{0}^{t} \underbrace{\left| \W u(\tau,x)  \right|}_{\leq M} \cdot \left |u(\tau,x)-\hat{u}(\tau,x) \right| \d \tau+\\
&\int_{0}^{t} \underbrace{\left|1-\hat{u}(\tau,x) \right|}_{\leq 1} \cdot \left| \W u(\tau,x)- \hat{\W}\hat{u}(\tau,x) \right| \d \tau+ \\
&\int_{0}^{t} \underbrace{\gamma(x)}_{ \leq M} \left | u(\tau,x)-\hat{u}(\tau,x) \right | \d \tau+ \int_{0}^{t} \underbrace{\hat{u}(\tau,x)}_{ \leq 1} \left | \gamma(x)-\hat{\gamma}(x)\right| \d \tau \\
\left |u(t,x)-\hat{u}(t,x) \right| \leq & \left |u(0,x)-\hat{u}(0,x) \right| +T \left| \gamma(x)-\hat{\gamma}(x) \right| \\
&+\int_{0}^{t} \underbrace{\left | \hat{\W} \left[u(\tau,x)-\hat{u}(\tau,x) \right]  \right|}_{ \leq M \left \|u(\tau)-\hat{u}(\tau) \right \|_1} \d \tau+ \\
& \int_{0}^{t} \left |  \left[\W-\hat{\W} \right]u(\tau,x)
 \right | \d \tau+2M \int_{0}^{t} \left |u(\tau,x)-\hat{u}(\tau,x)\right | \d \tau \\
 \|u(t)-\hat{u}(t) \|_{1} \leq& \|u(0)-\hat{u}(0) \|_{1}+T \|\gamma-\hat{\gamma} \|_{1}+\underbrace{\int_{0}^{t} \left \| \left(\W-\hat{\W} \right)u(t)\right \|_{1}}_{=H_{1,t} \leq H_{1,T} } + \\
 &3M \int_{0}^{t} \|u(\tau)-\hat{u}(\tau) \|_{1 } \d \tau.
\end{align*}
Finally, applying Gr\"onwall's Lemma gives
\begin{align*}
 \sup_{0 \leq t \leq T}\|u(t)-\hat{u}(t) \|_{1} \leq& \left( \|u(0)-\hat{u}(0) \|_{1}+T \|\gamma-\hat{\gamma}\|_{1}+H_{1,T} \right)e^{3MT}.
\end{align*}
\end{proof}

\begin{proof} (Lemma \ref{l:Szemeredi})

From \cite{Szemeredi_analysis}, for any $W$ and $\varepsilon>0$ there exists a $\tilde{W}$ such that $\|W-W' \|_{\square}<\varepsilon$, and $\tilde{W}$ is constant over some intervals $I_1', \dots, I_{K'}'$, with  $K' \leq 2^{ \left \lceil\frac{2}{\varepsilon^2} \right \rceil }$. Furthermore, when $W$ is discrete with partition $I_1,\dots, I_n$, we can choose $(I_k')_{k=1}^{K'}$ so that $(I_i)_{i=1}^{n}$ is a refinement of $(I_k')_{k=1}^{K'}$ if $n \geq 2^{ \left \lceil\frac{2}{\varepsilon^2} \right \rceil }.$

In the next step, we construct $u'(0).$ Since $0 \leq u(0,x) \leq 1$, we may use
$$I_k'':= \left \{ x \in [0,1]\, | \, \frac{k-1}{K''} \leq u(0,x)<\frac{k}{K''} \right \}$$
where $K'':= \left \lceil \frac{1}{\varepsilon} \right \rceil$; apart from $\frac{k-1}{K''} \leq u'(0,x) \leq \frac{k}{K''}$, the value of $u'(0,x)$ is arbitrary. Clearly $\|u(0)-u'(0) \|_1 \leq \varepsilon.$

Similarly, $0 \leq \gamma(x) \leq M$, hence, the same argument provides a $\gamma'$ with partition $I_1''', \dots, I_{K'''}'''$.

Take the smallest common refinement of $(I_k)_{k=1}^{K'}, (I_k)_{k=1}^{K''}, (I_k)_{k=1}^{K'''}$ and denote it by $(I_k)_{k=1}^{K}.$

We define $W'$ so that
$$ w_{kl}':=\frac{1}{\pi_k' \pi_l'} \int_{I_k} \int_{I_l}W(x,y) \d x \d y;$$
according to \cite{Szemeredi_analysis,ALON2003},
$$\|W-W' \|_{\square} \leq 2 \varepsilon.$$
 

We now specify the values of $u'(0)$ and $\gamma'$ as 
\begin{align*}
        u'_k(0):=& \frac{1}{\pi_k} \int_{I_k} u(0,x) \d x, \\
        \gamma_k':=& \frac{1}{\pi_k} \int_{I_k} \gamma(x) \d x,
    \end{align*}
resulting in
\begin{align*}
 \|p-p' \|_{\square}=\|W-W' \|_{\square}+\|u(0)-u'(0) \|_{1}+ \|\gamma -\gamma' \|_{1} \leq 4 \varepsilon.   
\end{align*}
\end{proof}

\subsection{Proof of Section \ref{s:results_reconstruction}}

\begin{proof}(Theorem \ref{t:reconstruction_characterisation})
Since the nullspace of $G(T_0)$ is composed of vectors $v \in \mathbb{R}$ such that \begin{align*}
    \int_{0}^{T_0} \langle v, z(t) \rangle^2 \d t =0, 
\end{align*} 
\ref{condreconb} $\iff$ \ref{condreconc} is trivial.

Now assume there exists a vector $0 \neq v \in \mathcal{N}.$

Set $u_i:= \frac{v_i}{\pi_i}$ and $\hat{w}_{ij}:=w_{ij}-\varepsilon u_iu_j.$ Clearly, $\hat{W}^{\intercal}=\hat{W}$ and we can choose $\varepsilon>0 $ to be small enough such that $m' \leq \hat{w}_{ij} \leq M'$ for any $0<m'<m \leq M <M'$ making $\hat{W} \sim \mathcal{P}_{m',M'}^{(n)}.$
\begin{align*}
  \eta^{(i)}_j=&\, \varepsilon u_i u_j \pi_j=\varepsilon u_i v_j  \Rightarrow \eta^{(i)}=\varepsilon u_i v ,\\
  \left\| \eta^{(i)} \right \|_{G(T_0)}^2=& \, \varepsilon^2 u_i^2 \underbrace{\|v \|_{G(T_0)}^2}_{=0}=0 \ \overset{\eqref{eq:H2_noice}}{\Rightarrow} H_{2,T_0}=0.
\end{align*}
$H_{2,T_0}=0$ and Lemma \ref{l:H1} implies $z(t)=\hat{z}(t)$ for all $0 \leq t \leq T_0$ while $\hat{W} \neq W,$ which means $W$ is not $q$-reconstructible.

Lastly, assume $\mathcal{N}= \{0 \}$ making $\| \cdot \|_{G(T_0)}^2$ a norm. Let $\hat{W}$ a network such that $\hat{z}(t)=z(t)$ on $t \in [0,T_0].$ 

Due to  \eqref{eq:A_observe} for $t \in [0,T_0]$ we also have $\mathcal{A}z(t)=\hat{\mathcal{A}}\hat{z}(t)= \hat{\mathcal{A}}z(t)$, implying $H_{2,T_0}=0.$ Consequently, \eqref{eq:H2_noice} yields
\begin{align*}
  \forall i \in V \ \|\eta^{(i)} \|_{G(T_0)}^2=0 \quad \Rightarrow \quad \forall i \in V \ \eta^{(i)}=0  \quad\Rightarrow \quad \hat{\mathcal{A}}=\mathcal{A} \quad \Rightarrow \quad \hat{W}=W,  
\end{align*}
making $W$ $q$-reconstructible.
\end{proof}

\begin{proof} (Theorem \ref{t:different_d})

Fix any parameter $q \in \mathcal{Q}_{m,M}^{(n)}$ with $\gamma=0$. We will use the notation $z^c(t):=1-z(t)$ for short.
\begin{align*}
 \frac{\d}{\d t}z_{i}^{c}(t)=&-\frac{\d}{\d t}z_i(t)=-z_i^c(t) \left(\mathcal{A}z(t) \right)_{i} \\
 z_{i}^c(t)=&\exp \left(-\int_{0}^{t} \left(  \mathcal{A} z(\tau)\right)_{i} \d \tau \right)z_{i}^{c}(0)
\end{align*}

We give a uniform upper bound on $z_{i}^c(t)$ first. As there is no recovery in the SI model, $z_{i}(t)$ is monotone increasing, so $z_{i}(t) \geq z_{i}(0) \geq m.$ Then
\begin{align*}
\left(\mathcal{A}z(t) \right)_i =&\sum_{j \in V} w_{ij} \pi_{j} z_{j}(t) \geq m^2 \sum_{j \in V} \pi_j=m^2,
\end{align*}
so
\begin{align}
\label{eq:ziexp}
z_{i}^{c}(t) \leq& e^{-m^2 t}.
\end{align}

Next,
\begin{align*}
\int_{0}^{t} \left(\mathcal{A} z(t)\right)_{i} \d \tau=&\int_{0}^{t} \left(\mathcal{A} \mathds{1} \right)_i \d \tau +\int_{0}^{t} \left(\mathcal{A} z^c(t)\right)_{i} \d \tau=d_it+O(1), \\
z_{i}^{c}(t)=&\, e^{-d_it+O(1)}z_{i}^{c}(0).
\end{align*}
($\mathds{1}$ denotes the all $1$ vector.) 

Take a $v \in \mathbb{R}^n$ such that for all $t \in [0,T_0]$ we have $\langle v, z(t) \rangle=0$. From Lemma \ref{l:analytic}, $\langle v, z(t) \rangle=0$ actually holds for any $t \geq 0$.

In particular, since $z(t)=1-z^c(t) \to \mathds{1}$ as $t \to \infty$, we must have $\langle v, \mathds{1} \rangle=0 $, which yields $\langle v, z^c(t) \rangle=0. $

Without loss of generality, we assume $d_1<\dots <d_n.$

Define $\zeta_{k}(t):=\sum_{i \geq k}z_{i}^c(t).$ Due to \eqref{eq:ziexp}, the slowest exponential decay dominates $\zeta_k(t)$ making
\begin{align*}
   \lim_{t \to \infty} \frac{z_{i}^{c}(t)}{\zeta_{k}(t)} =0 
\end{align*}
for all $i>k$. Therefore,
\begin{align*}
0=\lim_{t \to \infty} \frac{1}{\zeta_{1}(t)} \langle u, z^c(t) \rangle=\lim_{t \to \infty} \sum_{i \in V} u_{i} \frac{z_{i}^c(t)}{\zeta_1(t)}=u_1.    
\end{align*}
Inductively, we may assume $u_i=0$ for all $i \leq k.$
\begin{align*}
 0=\lim_{t \to \infty} \frac{1}{\zeta_{k+1}(t)} \langle u, z^c(t) \rangle=\lim_{t \to \infty} \sum_{i \geq k+1} u_i \frac{z_{i}^c(t)}{\zeta_{k+1}(t)}=u_{k+1},   
\end{align*}
showing that $u=0$. 

This means $W$ is $q$-reversible.
\end{proof}

\begin{proof} (Theorem \ref{t:reconstruction_bad}.)

Let $p,p' \in \mathcal{P}_{m,M}^{(n)}$ with corresponding trajectories $z(t), z'(t)$. From Lemma \ref{l:H1_graphon} and Remark \ref{eq:cut_norm_conv} we know that for any $\varepsilon>0$ there is a $\delta>0$ such that $\|p-p'\|_{\square}<\delta$ implies
\begin{align*}
    \sup_{0 \leq t \leq T} \|z(t)-z'(t) \|_{1,\pi}<\frac{\varepsilon}{2}.
\end{align*}

From Lemmas \ref{l:H1_graphon} and \ref{l:Szemeredi} we know that there exists an $ n_0$ such that for any $n \geq n_0$ and $p=(W,q) \in \mathcal{P}_{m,M}^{(n)}$ there is a block homogeneous $p'=(W',q') \in \mathcal{P}^{(n)}_{m,M}$  with $1 \leq K \leq K_{\max}(\varepsilon)$ 
 blocks such that $\|p-p' \|_{\square}<\delta.$

Without loss of generality, we may assume  $\delta<\frac{1}{16}m$ and $n_0 \geq 4 K_{\max}(\varepsilon)$.

We will modify $W'$ into $W''$ in such a way that $p'':=(\hat{W},q')$ is still block regular with the same partition and $\bar{a}_{kl}$ values,  while $W'$ and $\hat{W}$ have a non-vanishing cut distance. This also means that the corresponding trajectory $z''(t)$ is identical to $z'(t).$

Informally, the idea is that we split vertices into odd and even labels and double the weights between vertices of the same parity and zero opposite ones.

Let the size of the block be $|V_k|:=2s_k+r_k$ where $r_{k} \in \{0,1 \}$ and $s_k \in \mathbb{N}.$ We can further divide the blocks into $V_{k,0}, V_{k,1},V_{k,2}$ where $|V_{k,0}|=r_{k}$ and $|V_{k,1}|=|V_{k,2}|=s_k.$

Take $i \in V_{k}, j \in V_l $ If $i \in V_{k,0}$, then set $\hat{w}_{ij}:=w'_{ij}.$ For $ i \in V_{k,1}$, 
\begin{equation*}
        \hat{w}_{ij} := \begin{cases}
                        w_{ij}' \text{ if $j \in V_{l,0} $} \\
                        2w_{ij}' \text{ if $j \in V_{l,1}$} \\
                        0, \textit{  if $j \in V_{l,2}$}
                    \end{cases}
\end{equation*}
then for $i \in V_{k,2}$, similarly take
\begin{equation*}
        \hat{w}_{ij} := \begin{cases}
                        w_{ij}' \text{ if $j \in V_{l,0} $} \\
                        0 \text{ if $j \in V_{l,1}$} \\
                        2w_{ij}' \text{ if $j \in V_{l,2}$.}
                    \end{cases}
\end{equation*}

Clearly, for any $k,l$ and $i \in V_k$,
$$\sum_{j \in V_l}\hat{w}_{ij} \frac{1}{n}=\sum_{j \in V_l}w'_{ij} \frac{1}{n},$$ and $\pi_i=\frac{1}{n}$, so indeed $p''$ is block regular and $z''(t)=z'(t).$

Furthermore, take the sets $S:=\bigcup_{k=1}^{K} V_{k,1}.$
\begin{align*}
    \|\hat{W}-W' \|_{\square} \geq & \frac{1}{n^2} \sum_{i,j \in S} \hat{w}_{ij}-w'_{ij}=\frac{1}{n^2} \sum_{i,j \in S} w_{ij}'=\frac{1}{n^2} \sum_{k,l=1}^{K} \bar{w}_{kl}s_{k}s_{l} \\
    \geq & \frac{1}{n^2}\sum_{k,l=1}^{K} \bar{w}_{kl} \frac{|V_k|-1}{2} \frac{|V_l|-1}{2} \geq \frac{m}{4n^2} \sum_{k,l=1}^{K} \left(|V_k|-1\right) \left(|V_l| -1\right) \\
    &=\frac{1}{4}m \frac{n^2-2Kn+K^2}{n^2}\geq \frac{1}{4}m \left(1-\frac{2K_{\max}(\varepsilon)}{n} \right) \geq \frac{1}{8}m .
\end{align*}
From $\|W-W' \|_{\square} \leq \|p-p'\|_{\square}<\delta<\frac{1}{16}m$ it follows that $\|W-\hat{W} \|_{\square}  \geq \frac{1}{16}m. $

Finally, set $\hat{p}:= (\hat{W},q)$ with trajectory $\hat{z}(t).$ Then
\begin{align*}
\|p''-\hat{p} \|_{\square}=\|z(0)-z'(0) \|_{1, \pi}+\|\gamma-\gamma' \|_{1, \pi} \leq d_{\square}(p,p')<\delta,
\end{align*}
resulting in
\begin{align*}
\sup_{0 \leq t \leq T} \|z(t)-\hat{z}(t) \|_{1, \pi} \leq &   \sup_{0 \leq t \leq T} \|z(t)-z'(t) \|_{1, \pi}+\sup_{0 \leq t \leq T} \|z'(t)-z''(t) \|_{1, \pi}  \\
& +\sup_{0 \leq t \leq T} \|z''(t)-\hat{z}(t) \|_{1, \pi}<  \frac{\varepsilon}{2}+0+\frac{\varepsilon}{2}=\varepsilon.
\end{align*}

\end{proof}
\subsection{Proofs for Section \ref{s:results_prediction}}

We will prove Theorem \ref{t:uniform_prediction} by showing that the following slightly stronger statement also holds:

\begin{theorem}
\label{t:uniform_prediction_stronger}
Assume $ 0 \leq \hat{W}(x,y)=\hat{W}(y,x) \leq M$ ($i,j \in V$) a.e. Then
\begin{align}
 \forall \varepsilon>0 \ \exists \delta>0 \ \forall  p \in  \mathcal{P}_{0,M}: \ H_{2,T_0}<\delta \Rightarrow H_{2,T}<\varepsilon .
\end{align}
\end{theorem}

Before proving Theorem \ref{t:uniform_prediction_stronger}, we make some preparations. Define
\begin{align*}
 G_T(x,y):=&\int_0^T u(t,x)u(t,y) \d t, \\
 \mathbb{G}_Tf(x):=& \int_{0}^1 G_T(x,y)f(y) \d y, \\
 \|f \|_{G_T}^2:=&\langle \mathbb{G}_Tf,f \rangle= \int_0^{T} \langle f, u(t) \rangle^2 \d t.
\end{align*}
Clearly $\|f \|_{G_T}^2 \leq T \|f \|_2^2.$ Due to Lemma \ref{l:analytic}, $\|f \|_{G_{T_0}}^2=0 \iff \|f \|_{G_T}^2=0$ $(0<T_0 \leq T).$

Similarly to \eqref{eq:H2_noice}, take $\eta(x)(y):=W(x,y)-\hat{W}(x,y)$ for which we have $\W u(t,x)-\hat{\W} u(t,x)= \langle \eta(x),u(t) \rangle$, and
\begin{align}
\label{eq:H2_noice_graphon}
H_{2,T}^2= \int_{0}^{T} \int_{0}^{1} \langle \eta(x),u(t) \rangle^2 \d x \d t=\int_{0}^{1} \| \eta(x) \|_{G_T}^2 \d x.
\end{align}

Throughout the rest of this section, $\rightharpoonup$ denotes weak convergence in $L^2([0,1]).$

\begin{lemma}
\label{l:GT_weak_cont}
 Assume $f_N \rightharpoonup f $ with $\|f_N \|_2^2 \leq 1$. Then $\|f_N \|_{G_T}^2 \to \| f\|_{G_T}^2$.   
\end{lemma}

\begin{proof}(Lemma \ref{l:GT_weak_cont})

$\langle f_N,u(t) \rangle^2 \leq \|f_N \|_2^2\cdot \|u(t) \|_2^2 \leq 1$, so dominated convergence applies to show
\begin{align*}
\lim_{N \to \infty} \|f_N \|_{G_T}^2=\int_{0}^{T} \lim_{N \to \infty} \langle f_N, u(t) \rangle ^2 \d t=\int_{0}^T \langle f,u(t) \rangle^2\d t=\|f \|_{G_T}^2.   
\end{align*}
\end{proof}

\begin{lemma}
\label{l:lambda_eigenish}
Define
\begin{align*}
\lambda_N:= \sup_{\|f \|_2^2 = 1} \frac{\|f \|_{G_T}^2}{\|f \|_{G_{T_0}^2}+\frac{1}{N}}
\end{align*}
There exists an $f \in L^2([0,1])$ such that $\|f\|_2^2 \leq 1$ and
\begin{align*}
\lambda_N=\frac{\|f \|_{G_T}^2}{\|f \|_{G_{T_0}^2}+\frac{1}{N}}
\end{align*}
\end{lemma}

\begin{proof} (Lemma \ref{l:lambda_eigenish})
Clearly $\lambda_N \leq NT.$

There must be a sequence $(f_n) \subset L^2([0,1])$ with $\|f_n \|_2^2= 1$ such that
\begin{align*}
 \lambda_N=\lim_{n \to \infty} \frac{\|f_n \|_{G_T}^2}{\|f_n \|_{G_{T_0}^2}+\frac{1}{N}} 
\end{align*}
 
Since the unit ball is weakly compact in $L^2([0,1])$, it must have a weakly convergent subsequence  $f_{n_k} \rightharpoonup f$  with $\|f \|_2^2 \leq 1$, resulting in
\begin{align*}
 \lambda_N=\lim_{k \to \infty} \frac{\|f_{n_k} \|_{G_T}^2}{\|f_{n_k} \|_{G_{T_0}}^2+\frac{1}{N}} = \frac{\|f \|_{G_T}^2}{\|f \|_{G_{T_0}}^2+\frac{1}{N}}.
\end{align*}
\end{proof}

Since the value of $p \in \mathcal{P}_{0,M}$ varies, we emphasize that $\lambda_N$ and $G_{T}$ depends on $p$, writing $\lambda_N(p)$ and $G_T(p)$ respectively.

\begin{remark}
\label{r:norm_relabel}
Based on Remark \ref{r:u_relabel}:
\begin{align*}
\langle f, u(t) \rangle=&\int_{0}^{1}f(x)u(t,x) \d x=\int_{0}^{1} f^{\phi}(x) u^{\phi}(t,x) \d x=\langle f^{\phi},(u(t))^{\phi} \rangle,\\
\|f \|_{G_T(p)}^2=&\int_{0}^{T} \langle f, u(t) \rangle^2 \d t =\int_{0}^{T} \langle f^{\phi},(u(t))^{\phi} \rangle^2 \d t= \left\|f^{\phi} \right\|_{G_T\left(p^{\phi}\right)}^2.
\end{align*}
\end{remark}

\begin{lemma}
\label{l:chi}
Define
\begin{align*}
    \chi_N:=\sup_{p \in \mathcal{P}_{0,M}} \frac{1}{N} \lambda_N(p).
\end{align*}
Then $\chi_N \to 0$ as $N \to \infty.$
\end{lemma}

\begin{proof} (Lemma \ref{l:chi}) Clearly $\chi_N \leq T.$

Define $\theta:=\limsup_{N \to \infty} \chi_N$.

Indirectly assume $\theta>0$. Then there must be a subsequence $(N_k)$ such that 
\begin{align*}
   \chi_{N_k} \geq \frac{1}{2}\theta>0 \quad  \forall \, k. 
\end{align*}
Therefore there must be a $p_{N_k} \in \tilde{\mathcal{P}}_{0,M}$ such that
\begin{align*}
    \bar{\lambda}_{N_k}:=\frac{1}{N_k}\lambda_{N_k}\left(p_{N_k} \right)=\frac{\|f_{N_k} \|_{G_T\left(p_{N_k} \right)}^2}{N_k \|f_{N_k} \|_{G_{T_0}\left(p_{N_k} \right)}^2+1} \geq \frac{1}{4} \theta >0
\end{align*}
for some $f_{N_k} \in L^2([0,1]).$

Due to Remark \ref{r:norm_relabel}, we may take an arbitrary $\phi_{N_k} \in \Phi,$ making
\begin{align*}
  \bar{\lambda}_{N_k}:=\frac{\left\|f_{N_k}^{\phi_{N_k}} \right\|_{G_T\left(p_{N_k}^{\phi_{N_k}} \right)}^2}{N_k \left\|f_{N_k}^{\phi_{N_k}} \right\|_{G_{T_0}\left(p_{N_k}^{\phi_{N_k}} \right)}^2+1} . 
\end{align*}

Since $\tilde{\mathcal{P}}_{0,M}$ is compact we can choose a suitable subsequence $(N_k')$ and $\phi_{N_k'} \in \Phi$ such that $ p_{N_k'}^{\phi_{N_K'}} \overset{\| \cdot \|_{\square}}{ \to}p \in \mathcal{P}_{0,M}$. Furthermore, the unit ball of $L^2([0,1])$ is weakly compact and $ \left\|f_{N_k'}^{\phi_{N_k'}} \right\|_2= \left \|f_{N_k'} \right \|_2 \leq 1$, thus, for some further subsequence $(N_k'')$ we also have  $f_{N_k'"}^{\phi_{N_k''}} \rightharpoonup f$ with $\|f \|_2^2 \leq 1.$
Let the solutions of \eqref{eq:u} with parameters $p_{N_k'}^{\phi_{N_k''}}, p$ be $u_{N_k''}, u.$ 
Then $u_{N_k''}(t) \overset{L^1([0,1])}{ \to} u(t)$ uniformly in $t \in [0,T]$ due to Remark \ref{eq:cut_norm_conv}.

Notice
\begin{align*}
 &\left | \left\|f_{N_k''}^{\phi_{N_k''}} \right \|_{G_{T} \left(p_{N_k''}^{\phi_{N_k''}} \right)}^2-\left\|f_{N_k''}^{\phi_{N_K''}} \right\|_{G_{T} \left(p \right)}^2 \right|\leq\\
 & \int_{0}^{T} \left| \left\langle f_{N_k''}^{\phi_{N_k''}},u_{N_k''}(t)  \right\rangle^2 -\left\langle f_{N_k''}^{\phi_{N_k''}},u(t)  \right\rangle^2 \right|  \ \d t  \leq \\
& \int_{0}^{T} \left |\left\langle f_{N_k''}^{\phi_{N_k''}},u_{N_k''}(t)-u(t)  \right\rangle \right| \cdot \left| \left\langle f_{N_k''}^{\phi_{N_k''}},u_{N_k''}(t)+u(t)  \right\rangle\right| \d t \leq \\
& 2 \int_{0}^{T} \|u_{N_k''}(t)-u(t) \|_2 \d t \leq 2 T \sup_{ 0 \leq t \leq T}  \|u_{N_k''}(t)-u(t) \|_2 \overset{|u_{N_k''}(t,x)-u(t,x)| \leq 1 }{\leq} \\
&2T \sqrt{\sup_{ 0 \leq t \leq T}  \|u_{N_k''}(t)-u(t) \|_1 } \to 0,
\end{align*}
 which implies
\begin{align*}
\lim_{k \to \infty} \left \|f_{N_k''}^{\phi_{N_{k}''}}  \right\|_{G_{T} \left(p_{N_k''} \right)}^2=\lim_{k \to \infty} \left\|f_{N_k''}^{\phi_{N_k''}}  \right\|_{G_{T} \left(p \right)}^2=\|f \|_{G_T(p)}^2.
\end{align*}

Define 
\begin{align*}
   \bar{\lambda}:=\lim_{k \to \infty} \bar{\lambda}_{N_k''} =\lim_{k \to \infty} \frac{\left\|f_{N_k''}^{\phi_{N_k''}} \right\|_{G_T\left(p_{N_k''}^{\phi_{N_k''}} \right)}^2}{N_k'' \|f_{N_k''}^{\phi_{N_k''}} \|_{G_{T_0}\left(p_{N_k''}^{\phi_{N_k''}} \right)}^2+1}=\lim_{k \to \infty} \frac{\|f \|_{G_T(p)}^2}{N_k''\|f \|_{G_{T_0}(p)}^2+1}.
\end{align*}
If $\| f\|_{G_{T_0}}^2>0$ then the limit is $\bar{\lambda}=0.$

If $\| f\|_{G_{T_0}}^2=0 \iff \|f \|_{G_T}^{2}=0$, then $\bar{\lambda}=\|f \|_{G_T}^{2}=0.$

This, however, contradicts $\bar{\lambda} \geq \frac{1}{4} \theta >0$.
\end{proof}

Finally we can prove Theorem \ref{t:uniform_prediction_stronger} which in turn proves Theorem \ref{t:uniform_prediction}.

\begin{proof}(Theorem \ref{t:uniform_prediction_stronger}.)

$ \|\eta(x) \|_2 \leq \|\eta(x) \|_{\infty}  \leq M. $ for almost every $x \in [0,1]$.
Furthermore, $H_{2,T}^2$ can be expressed as
\begin{align*}
H_{2,T}^2=&\int_{0}^1 \|\eta(x) \|_{G_T}^2 \d x \leq \int_{0}^{1} \lambda_N  \left(\|\eta(x) \|_{G_{T_0}^2}+\frac{1}{N}\|\eta(x) \|_2^2 \right) \d x \\  
 \leq &\lambda_{N}\underbrace{\int_{0}^T \| \eta(x) \|_{G_{T_0}}^2  \d x}_{=H_{2,T_0}^2}+M^2 \frac{\lambda_N}{N} \leq NTH_{2,T_0}^2+M^2 \chi_N.
\end{align*}

Fix $\varepsilon>0$. Based on Lemma \ref{l:chi} we can choose a large enough $N$ such that $M^2 \chi_N<\frac{\varepsilon^2}{2}.$ Set $\delta:=\frac{\varepsilon}{\sqrt{2NT}}.$

From this we can say that $H_{2,T_0}<\delta$ implies $H_{2,T}<\varepsilon.$ Note that the choice of $\delta$ only depends on $\varepsilon$ and it is independent of the choice of $p \in \mathcal{P}_{0,M}.$
\end{proof}

\subsection{Proofs for Section \ref{s:discretemeasure}}
\begin{proof} (Lemma \ref{l:ztpi2s}).
\begin{align*}
&z(t)-z^*(t)=z(0)-z^*(0)+
\sum_{i:t_i\leq t} (z(t_i)-z(t_{i-1}))-(z^*(t_i)-z^*(t_{i-1}))+O(\Delta_4),
\end{align*}
so
\begin{align*}
&\|z(t)-z^*(t)\|_{\pi,2}\leq\\
&\|z(0)-z^*(0)\|_{\pi,2}+
\max_{1\leq i\leq J}\frac{1}{t_i-t_{i-1}}\|(z(t_i)-z(t_{i-1}))-(z^*(t_i)-z^*(t_{i-1}))\|_{\pi,2}
\\
&\times \sum_{i:t_i\leq t} (t_i-t_{i-1}) + O(\Delta_4) \leq
\Delta_1+\Delta_3 T_0 + O(\Delta_4)= O(\Delta).
\end{align*}
\end{proof}

\begin{proof} (Theorem \ref{t:numerror}).
We have
\begin{align*}
\nonumber
&|\alpha_j(t_i)-({\mathcal{A}}z(t_i))_j|=
O\Bigg(|z_j(t_i)-z^*_j(t_i)|+|\gamma_j-\gamma_j^*|+\\
&\frac{1}{t_{i+1}-t_i}|(z_j(t_i)-z_j(t_{i-1}))-(z_j^*(t_i)-z_j^*(t_{i-1}))|+\\
&+\left|\frac{1}{t_{i}-t_{i-1}}(z_j(t_i)-z_j(t_{i-1}))-\frac{\d}{\d t}z_j(t_i)\right|
\Bigg)=O(\Delta),
\end{align*}
from which
\begin{align}
\label{eq:alphaAz}
&\sum_{i=1}^J\sum_{j=1}^n |\alpha_j(t_i)-({\mathcal{A}}z(t_i))_j|^2\pi_j (t_i-t_{i-1})= O(\Delta^2).
\end{align}
We also have
\begin{align}
\label{eq:AzhatAz}
&\sum_{i=1}^J\sum_{j=1}^n |({\hat{\mathcal{A}}}z(t_i))_j-({\hat{\mathcal{A}}}z^*(t_i))_j|^2\pi_j (t_i-t_{i-1})= O(\Delta^2)
\end{align}
since $\hat{\mathcal{A}}$ is bounded in $\|.\|_{\pi,2}$.

Putting \eqref{eq:alphaAz} and \eqref{eq:AzhatAz} together, we obtain
\begin{align}
\nonumber
&(H_{2,T_0}^*(\hat W))^2=
\sum_{i=1}^J\sum_{j=1}^n |\alpha_j(t_i)-(\hat{\mathcal{A}}z^*(t_i))_j|^2\pi_j (t_i-t_{i-1})=\\
\label{eq:AzhatAz2}
&=\sum_{i=1}^J\sum_{j=1}^n |({\mathcal{A}}z(t_i))_j-({\hat{\mathcal{A}}}z(t_i))_j|^2\pi_j (t_i-t_{i-1})+ O(\Delta)=\\
\nonumber
&=(H_{2,T_0}(\hat W))^2+O(\Delta).
\end{align}
In the last step, the error of the numerical integration is order $O(\Delta_4)$.

(We also implicitly used $(a+b)^2\leq 2(a^2+b^2)$ as a version of the triangle inequality for squares; this affects the constant factor in $O(\cdot)$. Also, the formula $(x+\varepsilon)^2=x^2+2x\varepsilon+\varepsilon^2$ means the final error in \eqref{eq:AzhatAz2} is $O(\Delta)$ instead of $O(\Delta^2)$.)
\end{proof}

\bibliographystyle{abbrv}
\bibliography{mf}

\begin{thebibliography}{10}

\bibitem{ALON2003}
N.~Alon, W.~{de la Vega}, R.~Kannan, and M.~Karpinski.
\newblock Random sampling and approximation of {MAX-CSPs}.
\newblock {\em Journal of Computer and System Sciences}, 67(2):212--243, 2003.
\newblock Special Issue on STOC 2002.

\bibitem{Beaufort2022}
L.-B. Beaufort, P.-Y. Massé, A.~Reboulet, and L.~Oudre.
\newblock Network reconstruction problem for an epidemic reaction--diffusion
  system.
\newblock {\em Journal of Complex Networks}, 10(6):cnac047, 10 2022.

\bibitem{SISdyn}
J.-F. Delmas, D.~Dronnier, and P.-A. Zitt.
\newblock An infinite-dimensional {SIS} model, 2020.

\bibitem{switchover2021}
G.~Ódor, D.~Czifra, J.~Komjáthy, L.~Lovász, and M.~Karsai.
\newblock Switchover phenomenon induced by epidemic seeding on geometric
  networks.
\newblock {\em Proceedings of the National Academy of Sciences},
  118(41):e2112607118, 2021.

\bibitem{santiago}
N.~Gozzi, M.~Tizzoni, M.~Chinazzi, L.~Ferres, A.~Vespignani, and N.~Perra.
\newblock {Estimating the effect of social inequalities in the mitigation of
  COVID-19 across communities in Santiago de Chile}.
\newblock {\em Nature Communications}, 12, 2021.

\bibitem{hladky2024}
J.~Hladký, E.~K. Hng, and A.~M. Limbach.
\newblock Graphon branching processes and fractional isomorphism, 2024.
\newblock https://arxiv.org/abs/2408.02528.

\bibitem{mobile_epidemic}
M.~T. Islam, M.~Akon, A.~Abdrabou, and X.~Shen.
\newblock Modeling epidemic data diffusion for wireless mobile networks.
\newblock {\em Wireless Communications and Mobile Computing}, 14(7):745--760,
  2014.

\bibitem{Keliger2022}
D.~Keliger, I.~Horváth, and B.~Takács.
\newblock Local-density dependent markov processes on graphons with
  epidemiological applications.
\newblock {\em Stochastic Processes and their Applications}, 148:324--352,
  2022.

\bibitem{lovaszbook}
L.~Lov\'asz.
\newblock {\em Large Networks and Graph Limits}, volume~60 of {\em Colloquium
  Publications}.
\newblock American Mathematical Society, 2012.

\bibitem{Szemeredi_analysis}
L.~László and B.~Szegedy.
\newblock {Szemerédi’s Lemma for the Analyst}.
\newblock {\em Geometric and Functional Analysis}, 17:252--270, 04 2007.

\bibitem{reconstruction_entropy}
C.~Murphy, V.~Thibeault, A.~Allard, and P.~Desrosiers.
\newblock Duality between predictability and reconstructability in complex
  systems.
\newblock {\em Nature Communications}, 15, 05 2024.

\bibitem{reconstruction_mieghem2022}
B.~Prasse and P.~V. Mieghem.
\newblock Predicting network dynamics without requiring the knowledge of the
  interaction graph.
\newblock {\em Proceedings of the National Academy of Sciences},
  119(44):e2205517119, 2022.

\bibitem{cutnorm_inequality}
L.~Ruiz, L.~F.~O. Chamon, and A.~Ribeiro.
\newblock {Reply to 'Comments on Graphon Signal Processing'
  [arXiv:2310.14683]}, 2024.
\newblock https://arxiv.org/abs/2401.05326v1.

\end{thebibliography}

\end{document}